\documentclass{lmcs}


\usepackage{csquotes}
\usepackage{bussproofs}
\usepackage{mathtools}
\usepackage{listings}
\usepackage{hyperref}

\newcommand{\bbN}{\mathbb{N}}
\newcommand{\calA}{\mathcal{A}}
\newcommand{\calF}{\mathcal{F}}
\newcommand{\calQ}{\mathcal{Q}}
\newcommand{\calR}{\mathcal{R}}

\newcommand{\frakR}{\mathfrak{R}}

\newcommand{\0}{\mathbf{0}}
\newcommand{\1}{\mathbf{1}}
\newcommand{\2}{\mathbf{2}}
\newcommand{\3}{\mathbf{3}}
\newcommand{\4}{\mathbf{4}}
\newcommand{\5}{\mathbf{5}}
\newcommand{\K}{\mathbf{K}}

\newcommand{\case}[2]{\ensuremath{\textnormal{\textbf{case }} #1 \textnormal{\textbf{ of }} \langle #2 \rangle}}
\newcommand{\caseelse}[3]{\ensuremath{\textnormal{\textbf{case }} {#1} \textnormal{\textbf{ of }} \langle #3 \mid #2 \rangle}}

\newcommand{\coq}[1]{\textnormal{[\href{https://github.com/uds-psl/coq-library-undecidability/blob/70dfc56f33a6e4835281044e68aa68279989047e/theories/LambdaCalculus/Reductions/SSTS01\_to\_HOMbeta.v\##1}{\raisebox{-0.2em}{\includegraphics[width=0.6em]{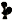}}}]}}
\newcommand{\coqextra}[1]{\textnormal{[\href{https://github.com/uds-psl/coq-library-undecidability/blob/70dfc56f33a6e4835281044e68aa68279989047e/theories/#1}{\raisebox{-0.2em}{\includegraphics[width=0.6em]{coq-logo-bw.pdf}}}]}}

\newcommand{\coqleft}[1]{\mbox{}\marginpar{\raggedleft\textnormal{\href{https://github.com/uds-psl/coq-library-undecidability/blob/70dfc56f33a6e4835281044e68aa68279989047e/theories/LambdaCalculus/Reductions/SSTS01\_to\_HOMbeta.v\##1}{\raisebox{-0.09em}{\includegraphics[width=0.53em]{coq-logo-bw.pdf}}}\hspace*{-0.8em}}}}
\newcommand{\coqextraleft}[1]{\mbox{}\marginpar{\raggedleft\textnormal{\href{https://github.com/uds-psl/coq-library-undecidability/blob/70dfc56f33a6e4835281044e68aa68279989047e/theories/#1}{\raisebox{-0.09em}{\includegraphics[width=0.53em]{coq-logo-bw.pdf}}}\hspace*{-0.8em}}}}

\newcommand{\ga}{\ensuremath{\iota}}

\theoremstyle{definition}
\newtheorem{problem}[thm]{Problem}

\begin{document}

\title[Undecidability of Higher-order beta-Matching]{Mechanized Undecidability of Higher-order beta-Matching (Extended Version)}

\author[A.~Dudenhefner]{Andrej Dudenhefner\lmcsorcid{0000-0003-1104-444X}}
\address{TU Dortmund University, Germany}
\email{andrej.dudenhefner@cs.tu-dortmund.de}

\begin{abstract}
Higher-order $\beta$-matching is the following decision problem: given two simply typed $\lambda$-terms, can the first term be instantiated to be $\beta$-equivalent to the second term?
This problem was formulated by Huet in the 1970s and shown undecidable by Loader in 2003 by reduction from $\lambda$-definability.

The present work provides a novel undecidability proof for higher-order $\beta$-matching, found in an effort to verify this result by means of a proof assistant.
Rather than starting from $\lambda$-definability, the presented proof encodes a restricted form of string rewriting as higher-order $\beta$-matching.
The particular approach is similar to Urzyczyn's approach showing undecidability of intersection type inhabitation.

The presented approach has several advantages.
First, the proof is simpler to verify in full detail due to the simple form of rewriting systems, which serve as a starting point.
Second, undecidability of the considered problem in string rewriting is already certified using the Rocq Prover.
As a consequence, we obtain a certified many-one reduction from the Halting Problem to higher-order $\beta$-matching.
Third, the presented approach identifies a uniform construction which shows undecidability of higher-order $\beta$-matching, $\lambda$-definability, and intersection type inhabitation.

The presented undecidability proof is mechanized in the Rocq Prover and contributed to the existing Coq Library of Undecidability Proofs.
\end{abstract}

\maketitle

\section{Introduction}

Higher-order $\beta$-unification in the simply typed $\lambda$-calculus is the following decision problem: given two simply typed $\lambda$-terms $M, N$, is there a substitution $S$ such that the instance $S(M)$ is $\beta$-equivalent to the instance $S(N)$?
Undecidability of higher-order $\beta$-unification was established by Huet~\cite{Huet73} in the 1970s, raising the question whether $\beta$-matching~\cite{Huet75} (the right-hand side term $N$ does not contain free variables) is decidable\footnote{Dowek~\cite{Dowek01} gives a comprehensive overview over unification and matching problems for the $\lambda$-calculus.}.
An equivalent presentation of higher-order $\beta$-matching (cf.\,Statman's range question~\cite{Statman82}) is: given a term~$F$ typed by the simple type $\sigma \to \tau$ and a term $N$ typed by the simple type $\tau$, is there a term~$M$ typed by the simple type $\sigma$ such that $F\,M$ is $\beta$-equivalent to $N$?

Decidability of higher-order $\beta$-matching was answered negatively\footnote{Not to be confused with the positive answer by Stirling~\cite{Stirling09} for higher-order $\beta\eta$-matching.} by Loader~\cite{Loader03} by reduction from $\lambda$-definability.
Loader introduces intricate machinery to formulate $\beta$-matching constraints which specify arbitrary finite functions.
Later, Joly~\cite{Joly05} refined Loader's result, shifting technical challenges to the underlying variant of $\lambda$-definability.
The intricate machinery renders verification of both approaches by means of a proof assistant quite challenging.

The present work presents a novel proof of the undecidability of higher-order $\beta$-matching, which is mechanized using the Rocq Prover~\cite{Coq_2023} (formerly known as the Coq proof assistant).

The mechanization leaves little room for ambiguities and potential errors, besides faithfulness of mechanized statements~\cite{pollack1998believe}.
This complements existing work on mechanized undecidability of higher-order $\beta$-unification~\cite{Spies020}.

The presented proof is not based on $\lambda$-definability;
rather, we consider a known rewriting problem in a restricted class of semi-Thue systems~\cite[Lemma~2]{Urzyczyn09} as a starting point.
The specific rewriting problem, referred to as $\0^+ \Rightarrow^* \1^+$, is: given a collection of rewrite rules of shape $ab \Rightarrow cd$, where $a,b,c,d$ are alphabet symbols, is there a non-empty sequence of $\0$s which can be transformed into a non-empty sequence of $\1$s?
As a consequence of the different starting point, the presented proof is simpler to verify in full detail and yields a concise mechanization.
The mechanization is incorporated into the existing Coq Library of Undecidability Proofs~\cite{CLUP20}, alongside the existing mechanization\footnote{The (Turing machine) Halting Problem is easily presented as Problem $\0^+ \Rightarrow^* \1^+$~\cite[Lemma~3.3]{DudenhefnerR19}.} of undecidability of the problem $\0^+ \Rightarrow^* \1^+$.

The main inspiration for the novel approach in the present work is Urzyczyn's undecidability proof for intersection type inhabitation~\cite{Urzyczyn09}.
Considering the relationship between the intersection type discipline and finite model theory~\cite{SalvatiMGB12}, the approach in the present work has an additional benefit:
it is uniformly applicable to prove undecidability of higher-order $\beta$-matching, intersection type inhabitation, and $\lambda$-definability.

The present work is an extended version of a conference article~\cite{Dudenhefner25}, and adds the following aspects.
\begin{itemize}
\item Further explanations throughout the textual passages.
\item Proof details apart from straightforward induction and case analysis.
\item Negative examples, showing what happens on the technical level when a given instance of the problem $\0^+ \Rightarrow^* \1^+$ has no solution.
\item Exemplified correspondence between intersection type inhabitation and higher-order $\beta$-matching.
\end{itemize}

\paragraph*{Paper organization} The present work is structured as follows: 
\begin{description}
\item[Section~\ref{sec:prelim}] Preliminaries for the simply typed $\lambda$-calculus, higher-order $\beta$-matching, and simple semi-Thue systems (including the undecidable Problem $\0^+ \Rightarrow^* \1^+$).
\item[Section~\ref{sec:undec}] Reduction from Problem $\0^+ \Rightarrow^* \1^+$ to higher-order $\beta$-matching.
\item[Section~\ref{sec:mech}] Overview over the mechanization in the Rocq Prover.
\item[Section~\ref{sec:uniform}] Applicability to intersection type inhabitation and $\lambda$-definability.
\item[Section~\ref{sec:concl}] Concluding remarks.
\end{description}

Statements in the digital version of the present work are linked to the corresponding mechanization, which is marked by the symbol\;\coqextra{LambdaCalculus/HOMatching.v}.

\newpage

\section{Preliminaries}\label{sec:prelim}
In this section we fix preliminaries and basic notation, following standard literature~\cite{Barendregt13book}.
\paragraph*{Higher-order $\beta$-Matching in the Simply Typed $\lambda$-Calculus}

The syntax of untyped $\lambda$-terms, called terms for brevity, is given in the following Definition~\ref{def:terms}.

\begin{defi}[$\lambda$-Terms]\label{def:terms}
\reversemarginpar\coqextraleft{L/L.v\#L8}%
\[
M, N ::= x \mid M \, N \mid \lambda x.M \quad \text{where } a, \ldots, z \text{ range over term variables}
\]
\end{defi}
Terms are considered modulo $\alpha$-equivalence (renaming of bound term variables).
Substitution of the term variable $x$ in the term $M$ by the term $N$ is denoted $M[x := N]$.
As usual, term application associates to the left, and we may group consecutive $\lambda$-abstractions.
We commonly refer to the term $\lambda x.x$ as $I$.

\begin{defi}[$\beta$-Reduction]
\reversemarginpar\coqextraleft{LambdaCalculus/Lambda.v\#L53}%
The relation $\to_\beta$ on terms is the contextual closure of $(\lambda x.M)\,N \to_\beta M[x := N]$.
\end{defi}
The $\beta$-equivalence relation $=_\beta$ is the reflexive, transitive, symmetric closure of $\to_\beta$.

In the simply typed $\lambda$-calculus we may assign to a term $M$ a simple type $\tau$ in type environment $\Gamma$, written as the judgment $\Gamma \vdash M : \tau$.
Similarly to prior work~\cite{Loader03}, one ground atom $\ga$ in the construction of simple types suffices for the negative result in the present work.
Definition~\ref{def:stlc} contains the rules (Var), ($\to$I), and ($\to$E) of the simple type system.

\begin{defi}[Simple Types with Ground Atom $\ga$]
\reversemarginpar\coqextraleft{LambdaCalculus/HOMatching.v\#L23}%
\[
\sigma, \tau ::= \ga \mid \sigma \to \tau
\]
\end{defi}
The arrow type constructor $\to$ associates to the right.

\begin{defi}[Type Environments]
\[
\Gamma ::= \{x_1 : \sigma_1, \ldots, x_n : \sigma_n\} \text{ where } x_i \neq x_j \text{ for } i \neq j
\]
\end{defi}

\begin{defi}[Simple Type System]\label{def:stlc}
\reversemarginpar\coqextraleft{LambdaCalculus/HOMatching.v\#L31}%
~\\[4pt]
\begin{tabular}{ccc}
{\RightLabel{\textnormal{(Var)}}
\AxiomC{$(x : \sigma) \in \Gamma$}
\UnaryInfC{$\Gamma \vdash x : \sigma$}
\DisplayProof}
&
{\RightLabel{\textnormal{($\to$I)}}
\AxiomC{$\Gamma, x : \sigma \vdash M : \tau$}
\UnaryInfC{$\Gamma \vdash \lambda x.M : \sigma \to \tau$}
\DisplayProof}
&
{\RightLabel{\textnormal{($\to$E)}}
\AxiomC{$\Gamma \vdash M : \sigma \to \tau$}
\AxiomC{$\Gamma \vdash N : \sigma$}
\BinaryInfC{$\Gamma \vdash M\,N : \tau$}
\DisplayProof}
\end{tabular}
\end{defi}
The following Example~\ref{xmp:derive} illustrates a type derivation in the simple type system.

\begin{exa}\label{xmp:derive}~\\
\begin{minipage}{\textwidth}
\centering
\smallskip
\AxiomC{$(f : \ga \to \ga) \in \{u : \ga, f : \ga \to \ga\}$}
\RightLabel{\textnormal{(Var)}}
\UnaryInfC{$ \{u : \ga, f : \ga \to \ga\} \vdash f : \ga \to \ga$}
\AxiomC{$(u : \ga) \in \{u : \ga, f : \ga \to \ga\}$}
\RightLabel{\textnormal{(Var)}}
\UnaryInfC{$ \{u : \ga, f : \ga \to \ga\} \vdash u : \ga$}
\RightLabel{\textnormal{($\to$E)}}
\BinaryInfC{$\{u : \ga, f : \ga \to \ga\} \vdash f\,u : \ga$}
\RightLabel{\textnormal{($\to$I)}}
\UnaryInfC{$\{u : \ga\} \vdash \lambda f.f\,u : (\ga \to \ga) \to \ga$}
\RightLabel{\textnormal{($\to$I)}}
\UnaryInfC{$\emptyset \vdash \lambda u.\lambda f.f\,u : \ga \to (\ga \to \ga) \to \ga$}
\DisplayProof
\end{minipage}
\end{exa}

Higher-order $\beta$-matching is the following typed unification problem, for which only one side is subject to instantiation. 

\begin{problem}[Higher-order $\beta$-Matching (${F}\,\mathsf{X} = N$)]\label{prb:hom}
\reversemarginpar\coqextraleft{LambdaCalculus/HOMatching.v\#L37}%
Given terms $F, N$ and simple types~$\sigma, \tau$ such that $\emptyset \vdash F : \sigma \to \tau$ and $\emptyset \vdash N : \tau$, is there a term $M$ such that $\emptyset \vdash M : \sigma$ and $F\,M =_\beta N$?
\end{problem}

Undecidability of higher-order $\beta$-matching is shown by Loader~\cite{Loader03} using a reduction from a variant of $\lambda$-definability.

\begin{thmC}[{\cite[Theorem~5.5]{Loader03}}]
Higher-order $\beta$-matching is undecidable.
\end{thmC}

For the remainder of the present work we use the term \emph{matching} in order to refer to higher-order $\beta$-matching.
Since simply typed terms are strongly normalizing and $\beta$-reduction is confluent~\cite{Barendregt13book}, it suffices in Problem~\ref{prb:hom} to consider terms $F, M, N$ in normal form.

Let us get familiar with matching by means of several illustrating examples.
The following Example~\ref{xmp:hom} illustrates a positive matching instance.

\begin{exa}\label{xmp:hom}
Consider the terms $F := \lambda x.\lambda y.x\,y\,I$ and $N := I$, for which we have $\emptyset \vdash F : (\ga \to (\ga \to \ga) \to \ga) \to (\ga \to \ga)$ and $\emptyset \vdash N : \ga \to \ga$.

The matching instance ${F}\,\mathsf{X} = N$ is solvable, including the solution $M := \lambda u.\lambda f.f\,u$.
More precisely, we have the following two properties:
\begin{itemize}
\item $\emptyset \vdash M : \ga \to (\ga \to \ga) \to \ga$ (Example~\ref{xmp:derive})
\item $F\,M =_\beta \lambda y.(\lambda u.\lambda f.f\,u)\,y\,I =_\beta \lambda y.I\,y =_\beta N$
\end{itemize}
\end{exa}

We want to encode certain functional behavior as a matching instance.
The following Example~\ref{xmp:ad-hoc} shows a naive approach to such an encoding and its limitations. 

\begin{exa}\label{xmp:ad-hoc}
Let us associate elements of the set $\{1, 2, 3\}$ with projections $\pi_1 := \lambda xyz.x$, $\pi_2 := \lambda xyz.y$, and $\pi_3 := \lambda xyz.z$ respectively.
For the term $G := \lambda h.\lambda x y z.h\,y\,z\,x$ we have
$G\,\pi_1 =_\beta \pi_2$, $G\,\pi_2 =_\beta \pi_3$, and $G\,\pi_3 =_\beta \pi_1$.
Therefore, $G$ realizes a finite function $f_G : \{1, 2, 3\} \to \{1, 2, 3\}$ such that $f_G(1) = 2$, $f_G(2) = 3$, and $f_G(3) = 1$.

Let $\kappa := \ga \to \ga \to \ga \to \ga$ be a simple type for which we have $\emptyset \vdash \pi_i : \kappa$ for $i \in \{1,2,3\}$.
Consider the matching instance $F\,\mathsf{X} = \pi_3$, where $F := \lambda t.t\,G\,\pi_1$, and for which we have $\emptyset \vdash F : ((\kappa \to \kappa) \to \kappa \to \kappa) \to \kappa$.
The intended \enquote{meaning} of this matching instance is: starting with the element $1$, repeatedly apply the function $f_G$ in order to construct the element $3$.

One solution for this instance is the term $\lambda f.\lambda s.f\,(f\,s)$ for which we have:
\[F\,(\lambda f.\lambda s.f\,(f\,s)) =_\beta f\,(f\,s)[f := G, s := \pi_1] =_\beta G\,\pi_2 =_\beta \pi_3\]
This solution follows the intended meaning of the underlying representation, constructing the element $f_G(f_G(1)) = 3$.
Another solution is $\lambda f.\lambda s.f\,(f\,(f\,(f\,(f\,s))))$, which corresponds to $f_G(f_G(f_G(f_G(f_G(1))))) = 3$.

Unfortunately, there are solutions to the above matching instance which behave differently.
One such solution is $\lambda f.\lambda s.\pi_3$, for which we also have $F\,(\lambda f.\lambda s.\pi_3) =_\beta \pi_3$.
In this case, the element $3$ is constructed \enquote{ad hoc}, with no reference to the provided arguments.
Another solution is the term $\lambda f.\lambda s.\lambda x y z.f\,s\,z\,z\,z$.
This solution exploits the exact representation of elements via projections, disregarding any intended meaning of the underlying representation.
\end{exa}

In the above Example~\ref{xmp:ad-hoc} we would like to exclude certain unintended solutions, in order to faithfully encode intended functional behavior.
Towards this aim, a technique how to restrict the shape of solutions is illustrated in the following Example~\ref{xmp:no-ad-hoc}.
We formulate an additional constraint that excludes the unintended solutions while preserving the intended ones.
The key observation is that ad hoc solutions do not use the provided arguments or contain additional $\lambda$-abstractions (going under the hood of the encoding via projections).
We can use a \emph{free variable}, to both enforce argument use and hinder the construction of additional $\lambda$-abstractions.

\begin{exa}[{\cite[Proposition~3.4]{Dowek01}}]\label{xmp:no-ad-hoc}
Consider a term $M$ in normal form such that $M\,I\,u =_\beta u$ where $u$ is a term variable, and $\emptyset \vdash M : (\kappa \to \kappa) \to \kappa \to \kappa$.
By case analysis on $M$ we have that $M = \lambda f.\lambda s.N$ for some term $N$ in normal form.
Furthermore, $\{f : \kappa \to \kappa, s : \kappa\} \vdash N : \kappa$ and $N[f := \lambda x.x, s := u] =_\beta u$.
Therefore, the term $N$ is not an abstraction.
By induction on the size of $N$ and case analysis of the normal form we have that $N = s$ or $N = f\,(\ldots(f\,s)\ldots)$.
Since the term $M$ contains exactly two $\lambda$-abstractions, it cannot be an ad hoc solution from Example~\ref{xmp:ad-hoc}.
\end{exa}

\begin{rem}\label{rem:typing-break}
Without the typing restriction $\emptyset \vdash M : (\kappa \to \kappa) \to \kappa \to \kappa$ in the above Example~\ref{xmp:no-ad-hoc} the term $M := \lambda f.\lambda s.(f\,f)\,s$ satisfies $M\,I\,u =_\beta u$.
\end{rem}

Combining Example~\ref{xmp:ad-hoc} with Example~\ref{xmp:no-ad-hoc}, we formulate in the following Example~\ref{xmp:no-ad-hoc-combo} a matching instance which faithfully encodes the desired functional behavior.

\begin{exa}\label{xmp:no-ad-hoc-combo}
Let $F := \lambda t.\lambda r.r\,(t\,G\,\pi_1)\,(\lambda u.t\,I\,u)$ and $N := \lambda r.r\,\pi_3\,(\lambda u.u)$ where the term $G = \lambda h.\lambda x y z.h\,y\,z\,x$ is from Example~\ref{xmp:ad-hoc}. We have
\begin{itemize}
\item $\emptyset \vdash F : ((\kappa \to \kappa) \to \kappa \to \kappa) \to (\kappa \to (\kappa \to \kappa) \to \ga) \to \ga$
\item $\emptyset \vdash N : (\kappa \to (\kappa \to \kappa) \to \ga) \to \ga$
\end{itemize}
The matching instance $F\,\mathsf{X} = N$ combines the matching instance from Example~\ref{xmp:ad-hoc} with the additional restriction from Example~\ref{xmp:no-ad-hoc}. 
Meaningful solutions such as $\lambda f.\lambda s.f\,(f\,s)$ and $\lambda f.\lambda s.f\,(f\,(f\,(f\,(f\,s))))$ from Example~\ref{xmp:ad-hoc} still solve $F\,\mathsf{X} = N$.
However, solutions with too many $\lambda$-abstractions such as $\lambda f.\lambda s.\pi_3$ or $\lambda f.\lambda s.\lambda x y z.f\,s\,z\,z\,z$ from Example~\ref{xmp:ad-hoc} do not.
\end{exa}

The following Remark~\ref{rem:eta-break} illustrates how the addition of $\eta$-reduction (contextual closure of $\lambda x.f\,x \to_\eta f$ where $x$ does not appear free in $f$) undermines the technique from Example~\ref{xmp:no-ad-hoc-combo} (as well as Loader's approach).

\begin{rem}\label{rem:eta-break}
Example~\ref{xmp:no-ad-hoc} demonstrates how to limit abstractions in solutions.
However, in the presence of $\eta$-reduction this does not work.
Consider the term $G := \lambda h.\lambda x y z.h\,y\,z\,x$ from Example~\ref{xmp:ad-hoc}.
The term $M := \lambda g.\lambda h.\lambda xyz.h\,(g\,\pi_1 x\,z\,y)\,y\,z$ solves the matching instance in Example~\ref{xmp:no-ad-hoc-combo} because:\\
$\arraycolsep=1.4pt\begin{array}[t]{lclcl}
M\,G\,\pi_1 &=_\beta& \big(\lambda g.\lambda h.\lambda xyz.h\,(g\,\pi_1 x\,z\,y)\,y\,z\big)\,G\,\pi_1
&=_\beta& \lambda xyz.\pi_{1}\,(G\,\pi_1 x\,z\,y)\,y\,z\\
&=_\beta& \lambda xyz.G\,\pi_1 x\,z\,y
&=_\beta& \lambda xyz.\pi_2\,x\,z\,y =_\beta \pi_{3}
\end{array}$\\
$\arraycolsep=1.4pt\begin{array}[t]{lclcl}
M\,I\,u &=_\beta& \big(\lambda g.\lambda h.\lambda xyz.h\,(g\,\pi_1 x\,z\,y)\,y\,z\big)\,I\,u
&=_\beta& \lambda xyz.u\,(I\,\pi_1 x\,z\,y)\,y\,z\\
&=_\beta& \lambda xyz.u\,(\pi_1 x\,z\,y)\,y\,z &=_\beta& \lambda xyz.u\,x\,y\,z \to_\eta^* u
\end{array}$\\
In particular, $\eta$-reduction allows for additional $\lambda$-abstractions in the solution, making the technique from Example~\ref{xmp:no-ad-hoc} unsuitable.
\end{rem}

The observation from Example~\ref{xmp:no-ad-hoc-combo} is generalized by Loader to encode arbitrary families of finite functions.
This results in undecidability of higher-order $\beta$-matching by reduction from a variant of $\lambda$-definability.
Loader's generalization is quite sophisticated, as it requires construction principles to restrict shapes of realizers of arbitrary higher-order finite functions.
In the present work, we focus on a family of higher-order finite functions which can be identified by inspection of intersection types occurring in the undecidability proof of intersection type inhabitation~\cite{Urzyczyn09,DudenhefnerR19} and their relationship to finite model theory~\cite{SalvatiMGB12}.
This leads to a simpler undecidability proof and reveals a connection between matching, intersection type inhabitation, and $\lambda$-definability.
The presented approach has similarities with Joly's~\cite{Joly05} refinement of Loader's proof.
Joly shifts the technical burden to a restricted $\lambda$-definability problem.
Instead, we avoid $\lambda$-definability altogether and use a rewriting problem in a class of \emph{simple} semi-Thue systems as a starting point.

\paragraph*{Simple Semi-Thue Systems}

A semi-Thue system~\cite[Definition~2.1.1]{BookOtto93} is a rewriting system on words, given by a finite set of rules (Definition~\ref{def:sts}).

\begin{defi}[Semi-Thue System]\label{def:sts}
A \emph{semi-Thue system} $\frakR$ over a finite alphabet $\calA$ is a finite set of rules $u \Rightarrow v$, where $u, v \in \calA^*$.
The associated rewriting relation $\Rightarrow_\frakR$ on $\calA^*$ is given by
$xuy \Rightarrow_\frakR xvy$ for $(u \Rightarrow v) \in \frakR$ and $x, y \in \calA^*$.
\end{defi}

A simple semi-Thue system (Definition~\ref{def:ssts}) is a semi-Thue system of restricted shape, introduced by Urzyczyn~\cite{Urzyczyn09} in order to show undecidability of intersection type inhabitation.

\begin{defi}[Simple Semi-Thue System]\label{def:ssts}
\coqextraleft{StringRewriting/SSTS.v\#L23}%
A semi-Thue system $\frakR$ over an alphabet $\calA$ is \emph{simple}, if each rule has the shape $ab \Rightarrow cd$ for $a,b,c,d \in \calA$.
\end{defi}

The reflexive, transitive closure of the rewriting relation for a given simple semi-Thue system~$\frakR$ is denoted $\Rightarrow_\frakR^*$.
For arbitrary simple semi-Thue systems it is undecidable whether some non-empty sequence of $\0$s can be transformed into a non-empty sequence of $\1$s.

\begin{problem}[$\0^+ \Rightarrow^* \1^+$]\label{prb:ssts01}
\coqextraleft{StringRewriting/SSTS.v\#L37}%
Given a simple semi-Thue system $\frakR$, does $\0^n \Rightarrow^*_\frakR \1^n$ hold for some $n > 0$?
\end{problem}

\begin{thmC}[{\cite[Lemma~3.3]{DudenhefnerR19}}]\label{thm:ssts_undec}
\coqextraleft{StringRewriting/SSTS\_undec.v\#L21}%
Problem $\0^+ \Rightarrow^* \1^+$ is undecidable.
\end{thmC}
\begin{proof}[Proof Sketch]
Reduction from the Turing machine halting problem from the empty tape.
Rule shape $ab \Rightarrow cd$ suffices to represent any Turing machine transition function.
For example, the transition from state $p$ to state $q$ reading symbol $x$, writing symbol $y$, and moving to the right to symbol $z$ can be represented as $(p, x)(\_, z) \Rightarrow (\_, y)(q, z)$.
The alphabet of the corresponding simple semi-Thue system includes symbols $\0$, $\1$, and pairs of state (including a special state \enquote{$\_$}) and tape symbol.
The Turing machine head is positioned at the pair for which the first projection is not the special state \enquote{$\_$}.
Some care needs to be taken with respect to correct word initialization and finalization after a halting state is reached.
The length $n > 0$ of starting word $\0^n$ is the space (length of the tape) required for the given Turing machine to reach a halting configuration.
\end{proof}

The following Example~\ref{xmp:ssts} shows a positive instance of Problem $\0^+ \Rightarrow^* \1^+$, which we will use throughout Section~\ref{sec:undec} to illustrate the reduction from Problem $\0^+ \Rightarrow^* \1^+$ to matching.

\begin{exa}\label{xmp:ssts}
Let $\frakR  := \{\0\0 \Rightarrow \2\2, \0\2 \Rightarrow \1\1, \2\0 \Rightarrow \1\1\}$ be a simple semi-Thue system over the alphabet $\{\0,\1,\2\}$.
We have $\0\0\0\0 \Rightarrow_\frakR \0\2\2\0 \Rightarrow_\frakR \1\1\2\0 \Rightarrow_\frakR \1\1\1\1$. As a side note, we have $\0^n \;{\not\Rightarrow}_\frakR^*\; \1^n$ for $n \in \{1, 2, 3\}$.
\end{exa}

Complementarily to Example~\ref{xmp:ssts}, the following Example~\ref{xmp:neg-ssts} shows a negative instance of Problem $\0^+ \Rightarrow^* \1^+$.

\begin{exa}\label{xmp:neg-ssts}
Let $\frakR  := \{\0\0 \Rightarrow \1\0, \0\1 \Rightarrow \1\1\}$ be a simple semi-Thue system over the alphabet $\{\0,\1\}$.
After any number of rewriting steps starting with $\0^{n+1}$ the last symbol is~$\0$, because it cannot be changed by any rule in $\frakR$.
Therefore, for any $n > 0$ we have $\0^n \not\Rightarrow_\frakR^* \1^n$.
\end{exa}

\begin{rem}
Simple semi-Thue systems were introduced by Urzyczyn to establish the undecidability result for intersection type inhabitation~\cite{Urzyczyn09}.
The exact presentation of Problem $\0^+ \Rightarrow^* \1^+$ was introduced later as a starting point in a refinement~\cite[Lemma~4.4]{DudenhefnerR19} of Urzyczyn's undecidability result for intersection type inhabitation~\cite{Urzyczyn09}.
Undecidability of Problem $\0^+ \Rightarrow^* \1^+$ is mechanized as part of Coq Library of Undecidability Proofs~\cite{CLUP20}, making it a good starting point for mechanized undecidability results.
\end{rem}

\section{Undecidability of Higher-order $\beta$-Matching}\label{sec:undec}
In this section we develop our main result (Theorem~\ref{thm:ssts_to_hom}): a reduction from the rewriting problem $\0^+ \Rightarrow^* \1^+$ to higher-order $\beta$-matching.

For the remainder of the section we fix the simple semi-Thue system $\frakR := \{R_1, \ldots, R_L\}$ with $L > 0$ rules over the finite alphabet $\{\0, \1, \2, \ldots, \mathbf{K}\}$.
Our approach is to construct simply typed terms which capture the two main aspects of the rewriting problem $\0^+ \Rightarrow^* \1^+$: the search for a sufficiently long starting sequence of $\0$s, and the individual rewriting steps to the desired sequence of $\1$s.

The remainder of the present section is structured as follows.
First, we fix basic notation, encoding, and properties of rewriting in the system $\frakR$.
Second, we restrict the shape of potential solutions for the constructed matching instance, similarly to Example~\ref{xmp:no-ad-hoc}.
Third, for solutions of restricted shape we capture the functional properties of the rewriting problem $\0^+ \Rightarrow^* \1^+$.

\subsection*{Notation}

We introduce four additional symbols in extended alphabet $\calA := \{\0, \1, \ldots, \mathbf{K}\} \cup \{\$, \bullet, \top, \bot\}$.
We represent a symbol $i \in \calA$ as the projection $\pi_i := \lambda s_\0 s_\1 \ldots s_\K s_\$ s_\bullet s_\top s_\bot.s_i$ typed by the simple type $\kappa := \underbracket{\ga \to \ldots \to \ga}_{|\calA|~\text{times}} \to \ga$.
The intended meaning of the additional symbols in the extended alphabet is as follows.
The symbol $\$$ marks the beginning of a word for word expansion,
$\bullet$ serves as symbol different from $\0$ and $\1$ for case analysis,
$\top$ serves as a constructible symbol, and $\bot$ serves as a non-constructible symbol.

For readability, we use the following $\textbf{case}$ notation to match individual symbols.

\begin{defi}[\textbf{case}]\label{def:case}
For $k \in \bbN$, distinct $i_1, \ldots, i_k \in \calA$, and terms $M, M_1, \ldots, M_k$:
\begin{align*}
\caseelse{x}{i_1 \mapsto M_1 \mid \ldots \mid i_k \mapsto M_k}{M} := &~x\,N_\0\,N_\1 \ldots N_\K\,N_\$\,N_\bullet\,N_\top\,N_\bot\\
&\text{ where } N_i = \begin{cases} M_j & \text{if } i = i_j\\M & \text{otherwise}\end{cases}
\end{align*}
\end{defi}

\begin{lem}
Let $k \in \bbN$, let $i_1, \ldots, i_k \in \calA$ be distinct, and let $M, M_1, \ldots, M_k$ be some terms.
For $j \in \calA$ we have 
$\caseelse{\pi_j}{i_1 \mapsto M_1 \mid \ldots \mid i_k \mapsto M_k}{M} =_\beta \begin{cases} M_l \text{ if } j = i_l \\ M \text{ else} \end{cases}$
\end{lem}

A particular term $\delta_i$ for $i \in \calA$, which we use commonly is:
\[\delta_i := \lambda x.\lambda s_\0 s_\1 \ldots s_\K s_\$ s_\bullet s_\top s_\bot.\caseelse{x}{\top \mapsto s_i}{s_\bot}\]
We have $\emptyset \vdash \delta_i : \kappa \to \kappa$, and the following Lemma~\ref{lem:delta} specifies the behavior of $\delta_i$.

\begin{lem}\label{lem:delta}
For $i, j \in \calA$ such that $j \neq \top$ we have
$\delta_i\,\pi_\top =_\beta \pi_i$ and $\delta_i\,\pi_j =_\beta \pi_\bot$.
\end{lem}

Intuitively, $\delta_i$ is $\pi_i$ when provided with the easily constructed $\pi_\top$.
However, it is essential in the formulation of constraints restricting solution shape that $\delta_i$ is assigned the same simple type $\kappa \to \kappa$ as the term $I$.
In particular, we will use the term $I$ (similarly to Example~\ref{xmp:no-ad-hoc}) of type $\kappa \to \kappa$  in combination with a free variable $u$ of type $\kappa$ in order to restrict solution shape.

\subsection*{Syntactic Constraints}
We identify the shape of \enquote{well-formed} terms, suitable to represent rewriting.
In the following Definition~\ref{def:ring} terms in the set $\calQ_m$ capture consecutive rule application for a word of length $m + 1$, ending in the word $\1^{m+1}$ (represented by $z_\1 \in \calQ_m$).
The subterm $r_i\,p_j$ (and $r_i\,(\lambda w.p_j\,w)$) for $i \in \{1, \ldots, L\}$ and $j \in \{1, \ldots, m\}$ indicates an application of the rule $R_i$ at position $j$.
Additionally, terms in the set $\calR_m$ capture consecutive increase of word length starting with $m+1$, and initialization with $\0$s before rewriting (represented by $z_\0\,N\,M \in \calR_m$ for $M \in \calQ_m$).
Specifically, the subterm $z_\star\,N\,(\lambda p_{m+1}.M)$ introduces an additional bound variable $p_{m+1}$ in order to argue about rule application at position $m+1$ in the longer word.
Consequently, terms in $\calR_1$ represent witnesses for an arbitrary expansion of a word starting with length $2$, followed by initialization with $\0$s, and consecutive rule application (potentially ending in $\1$s).

\begin{defi}[Sets $\calQ_m$, $\calR_m$ of Terms]\label{def:ring}
For $m > 0$, let $\calQ_m$ and $\calR_m$ be the smallest sets of terms satisfying the following rules:
\begin{itemize}
\item $z_\1 \in \calQ_m$
\item if $M \in \calQ_m$ then $(r_i\,p_j\,M) \in \calQ_m$ for $i \in \{1, \ldots, L\}$ and $j \in \{1, \ldots, m\}$
\item if $M \in \calQ_m$ then $(r_i\,(\lambda w.p_j\,w)\,M) \in \calQ_m$ for $i \in \{1, \ldots, L\}$ and $j \in \{1, \ldots, m\}$
\item if $M \in \calQ_m$ then $(z_\0\,N\,M) \in \calR_m$ for any term $N$
\item if $M \in \calR_{m+1}$ then $(z_\star\,N\,(\lambda p_{m+1}.M)) \in \calR_m$ for any term $N$
\end{itemize}
\end{defi}

\begin{rem}\label{rem:R1}
Terms in $\calR_1$ are inspired by inhabitants in a refinement~\cite[Lemma~4.4]{DudenhefnerR19} of Urzyczyn's undecidability result for intersection type inhabitation~\cite{Urzyczyn09}.
\end{rem}

Free variables occurring in terms in $\calQ_m$ and $\calR_m$ are assigned simple types according to the following type environment $\Gamma_m$, such that terms in $\calQ_m$ and $\calR_m$ can be assigned the simple type $\kappa$.

\begin{defi}[Type Environment $\Gamma_m$]\label{def:Gamma}
\coqleft{L459}%
For $m > 0$ let
\begin{align*}
\Gamma_m &:= \{ \begin{aligned}[t]
&z_\1 : \kappa, z_\0 : (\kappa \to \kappa) \to \kappa \to \kappa, z_\star : (\kappa \to \kappa) \to ((\kappa \to \kappa) \to \kappa) \to \kappa,\\
&p_1 : \kappa \to \kappa, \ldots, p_m : \kappa \to \kappa,\\
&r_1 : (\kappa \to \kappa) \to \kappa \to \kappa, \ldots, r_L : (\kappa \to \kappa) \to \kappa \to \kappa\}
\end{aligned}
\end{align*}
\end{defi}

Similarly to Example~\ref{xmp:no-ad-hoc}, we formulate typed terms (Definition~\ref{def:H}) and $\beta$-equivalence constraints characterizing members of $\calQ_m$ (Lemma~\ref{lem:ring2}) and $\calR_m$ (Lemma~\ref{lem:ring1}).

\begin{defi}[Typed Terms $H_\star$, $H_\0$, $H_R$]\label{def:H}
\begin{align*}
H_\star := &\lambda h.\lambda g.\lambda s_\0 s_\1 \ldots s_\K s_\$ s_\bullet s_\top s_\bot. \caseelse{g\,\delta_\bullet}{\$ \mapsto s_\$}{s_\bot}\\
H_\0 := &\lambda h.\lambda x. \lambda s_\0 s_\1 \ldots s_\K s_\$ s_\bullet s_\top s_\bot.\caseelse{x}{\1 \mapsto s_\$}{s_\bot}\\
H_R := &\lambda h.\lambda x. \lambda s_\0 s_\1 \ldots s_\K s_\$ s_\bullet s_\top s_\bot.\caseelse{h\,\pi_\top}{\bullet \mapsto \caseelse{x}{\1 \mapsto s_\1}{s_\bot} }{s_\bot}
\end{align*}
\[
\emptyset \vdash H_\star : \Gamma_m(z_\star) \qquad\quad
\emptyset \vdash H_\0 : \Gamma_m(z_\0) \qquad\quad
\emptyset \vdash H_R : \Gamma_m(r_i) \text{ for } i \in \{1, \ldots, L\}
\]
\end{defi}

We introduce substitutions $S_F$ and $S_H$ acting on the term variables $z_\star, z_\1, z_\0, r_1, \ldots, r_L$, which occur in terms in $\calQ_m$ and $\calR_m$.
In particular, the term $I$ along with the free variable~$u$ is used in $S_F$ analogously to Example~\ref{xmp:no-ad-hoc} in order to hinder the use of unwanted $\lambda$-abstractions. 

\begin{defi}[Substitutions $S_F$, $S_H$]\label{def:SF_SH}
~\\
$\begin{array}{llll}
S_F(z_\star) := \lambda h.\lambda g.g\,I ~&
S_F(z_\1) := u ~&
S_F(z_\0) := \lambda h.I ~&
S_F(r_j) := I \text{ for } j \in \{1, \ldots, L\}\\
S_H(z_\star) := H_\star &
S_H(z_\1) := \pi_\1 &
S_H(z_\0) := H_\0 &
S_H(r_j) := H_R \text{ for } j \in \{1, \ldots, L\}
\end{array}$
\end{defi}

\begin{lem}\label{lem:ring2}
\coqleft{L1687}%
For $m > 0$, if a term $M$ is in normal form such that $\Gamma_m \vdash M : \kappa$,\\
$S_F(M)[p_1 {\,:=\,} I, \ldots, p_m {\,:=\,} I] =_\beta u$, and
$S_H(M)[p_1 := \delta_\bullet, \ldots, p_m := \delta_\bullet] =_\beta \pi_{\1}$,
then $M \in \calQ_m$.
\end{lem}

\begin{proof}
Induction on the size of $M$ and case analysis of the normal form, observing that $S_F(z_\1) = u$ is the only occurrence of $u$ in $S_F$, and $S_H(z_\1)$ and $S_H(r_j)$ for $j \in \{1, \ldots, L\}$ are the only ways to construct $\pi_\1$.
\end{proof}

\begin{lem}\label{lem:ring1}
\reversemarginpar\coqleft{L1825}%
For $m > 0$, if a term $M$ is in normal form such that $\Gamma_m \vdash M : \kappa$,\\
$S_F(M)[p_1 {\,:=\,} I, \ldots, p_m {\,:=\,} I] =_\beta u$, and
$S_H(M)[p_1 := \delta_\bullet, \ldots, p_m := \delta_\bullet] =_\beta \pi_{\$}$,
then $M \in \calR_m$.
\end{lem}

\begin{proof}
Induction on the size of $M$, case analysis of the normal form, and Lemma~\ref{lem:ring2}, observing that $S_H(z_\0)$ and $S_H(z_\star)$ are the only ways to construct $\pi_\$$.
\end{proof}

As a consequence of the above Lemma~\ref{lem:ring1}, the following Theorem~\ref{thm:shape} presents $\beta$-equivalence constraints which suffice to restrict the shape of terms under consideration.

\begin{thm}\label{thm:shape}
\reversemarginpar\coqleft{L1932}%
If a term $M$ is in normal form such that\\$\emptyset \vdash M : \Gamma_1(r_1) \to \cdots \to \Gamma_1(r_L) \to \Gamma_1(z_\0) \to \Gamma_1(z_\1) \to \Gamma_1(z_\star) \to \Gamma_1(p_1) \to \kappa$,\\
$M\,I \ldots I\,(\lambda h.I)\,u\,(\lambda h.\lambda g.g\,I)\,I =_\beta u$, and
$M\,H_R \ldots H_R\,H_\0\,\pi_\1\,H_\star\,\delta_\bullet =_\beta \pi_{\$}$,\\
then $M = \lambda r_1 \ldots r_L.\lambda z_\0 z_\1 z_\star p_1.N$ for some term $N$ such that $N \in \calR_1$.
\end{thm}

\begin{proof}
Induction on the size of $M$, case analysis of the normal form, and Lemma~\ref{lem:ring1}.
\end{proof}

\begin{exa}\label{xmp:H}
Assume $\frakR = \{\0\0 \Rightarrow \2\2, \0\2 \Rightarrow \1\1, \2\0 \Rightarrow \1\1\}$ over the alphabet $\{\0,\1,\2\}$ from Example~\ref{xmp:ssts}.
Let $N := r_1\,p_2\,(r_2\,p_1\,(r_3\,p_3\,z_\1))$ and
$M := z_\star\,p_1\,(\lambda p_2.z_\star\,p_2\,(\lambda p_3.z_\0\,p_3\,N))$.
We have $N \in \calQ_3$ and $M \in \calR_1$.
In congruence with Theorem~\ref{thm:shape} we have:\\
$\begin{array}{ll}
(\lambda r_1 r_2 r_3.\lambda z_\0 z_\1 z_\star p_1.M)\,I\,I\,I\,(\lambda h.I)\,u\,(\lambda h.\lambda g.g\,I)\,I &=_\beta u\\
(\lambda r_1 r_2 r_3.\lambda z_\0 z_\1 z_\star p_1.M)\,H_R\,H_R\,H_R\,H_\0\,\pi_\1\,H_\star\,\delta_\bullet &=_\beta \pi_{\$}
\end{array}$\\
While Theorem~\ref{thm:shape} only establishes \enquote{well-formedness}, the term $M$ has an intended meaning:
An initial word of length $2$ is expanded twice (using $z_\star$) to a word of length $4$ and initialized to $\0$s (using $z_\0$).
The introduced variables $p_2$ and $p_3$ are used to address positions in the longer word.
The intended meaning of $N$ is that the first rule (using $r_1$) is applied at position $2$ (using $p_2$), followed by the second rule at position $1$, and third rule at position $3$ accordingly.
The resulting word contains only $\1$s (indicated by $z_\1$).
Overall, this corresponds to $\0\0\0\0 \Rightarrow_\frakR \0\2\2\0 \Rightarrow_\frakR \1\1\2\0 \Rightarrow_\frakR \1\1\1\1$.
\end{exa}

The following Example~\ref{xmp:no-H} shows that we still need to represent actual rewriting steps, as the terms $H_R$ do not distinguish distinct rewriting rules.

\begin{exa}\label{xmp:no-H}
Assume $\frakR = \{\0\0 \Rightarrow \1\0, \0\1 \Rightarrow \1\1\}$ over the alphabet $\{\0,\1\}$ from Example~\ref{xmp:neg-ssts}.
Let $N := r_1\,p_1\,(r_2\,p_2\,z_\1)$ and
$M := z_\star\,p_1\,(\lambda p_2.z_\0\,p_2\,N)$.
We have $N \in \calQ_2$ and $M \in \calR_1$.
In congruence with Theorem~\ref{thm:shape} we have:\\
$\begin{array}{ll}
(\lambda r_1 r_2.\lambda z_\0 z_\1 z_\star p_1.M)\,I\,I\,(\lambda h.I)\,u\,(\lambda h.\lambda g.g\,I)\,I &=_\beta u\\
(\lambda r_1 r_2.\lambda z_\0 z_\1 z_\star p_1.M)\,H_R\,H_R\,H_\0\,\pi_\1\,H_\star\,\delta_\bullet &=_\beta \pi_{\$}
\end{array}$\\
While the term $M$ represents a valid word expansion to length $3$ and initialization to $\0\0\0$, it does not represent a valid sequence of rewriting steps.
While the first step $\0\0\0 \Rightarrow_\frakR \1\0\0$ is represented correctly by using the subterm $r_1\,p_1$, the subterm $r_2\,p_2\,z_\1$ means that next the rule $\0\1 \Rightarrow \1\1$ is applied at position $2$, which is invalid.
\end{exa}

Having only \enquote{well-formed} terms to consider (cf. Example~\ref{xmp:ad-hoc} and Example~\ref{xmp:no-ad-hoc}), next we focus on the functional properties of rewriting.

\subsection*{Semantic Constraints}

We formulate typed terms (Definition~\ref{def:G}) and $\beta$-equivalence constraints characterizing word expansion (Lemma~\ref{lem:R_ssts_sem}) and rewriting (Lemma~\ref{lem:Q_ssts_sem}).
The presented terms are programs which realize the intended meaning (Example~\ref{xmp:H}) of \enquote{well-formed} terms in $\calQ_m$ and $\calR_m$.
Specifically, $G_\star$ realizes word expansion, $G_\0$ realizes initialization with $\0$s, $G_{ab \Rightarrow cd}$ realizes rule $ab \Rightarrow cd$ application, and $G_j^i$ controls the effect at position $i$ for rule application at position $j$.
A~word of length $m$ is represented by $m + 1$ right-hand sides of $\beta$-equivalences.

\begin{defi}[Typed Terms $G_\star$, $G_\0$, $G_{ab \Rightarrow cd}$, $G_i^j$]
\label{def:G}
\begin{align*}
G_\star := &\lambda h.\lambda g.\lambda s_\0 s_\1 \ldots s_\K s_\$ s_\bullet s_\top s_\bot.\case{h\,\pi_\top}{s_\bot\\
&\hspace{1.0em} \mid \bullet \mapsto \caseelse{g\,\delta_\bullet}{\0 \mapsto s_\0 \mid \$ \mapsto \caseelse{g\,\delta_\0}{\1 \mapsto s_\$}{s_\bot}}{s_\bot}\\
&\hspace{1.0em} \mid \0 \mapsto \caseelse{g\,\delta_\1}{\0 \mapsto s_\1}{s_\bot}\\
&\hspace{1.0em} \mid \1 \mapsto \caseelse{g\,\delta_\bullet}{\0 \mapsto s_\0}{s_\bot}}\\
G_\0 := &\lambda h.\lambda x. \lambda s_\0 s_\1 \ldots s_\K s_\$ s_\bullet s_\top s_\bot.\case{h\,\pi_\top}{s_\bot\\
&\quad\mid \bullet \mapsto \caseelse{x}{\0 \mapsto s_\0 \mid \1 \mapsto s_\$}{s_\bot}\\
&\quad\mid\0 \mapsto \caseelse{x}{\0 \mapsto s_\1}{s_\bot}\\
&\quad\mid\1 \mapsto \caseelse{x}{\0 \mapsto s_\0 }{s_\bot}}\\
G_{ab \Rightarrow cd} := &\lambda h.\lambda x.\lambda s_\0 s_\1 \ldots s_\K s_\$ s_\bullet s_\top s_\bot.\case{h\,\pi_\top}{s_\bot\\
&\quad\mid\bullet \mapsto x\,s_\0\,s_\1\, \ldots s_\K\, s_\$\, s_\bullet\, s_\top\, s_\bot\\
&\quad\mid\0 \mapsto \caseelse{x}{d \mapsto s_b}{s_\bot}\\
&\quad\mid\1 \mapsto \caseelse{x}{c \mapsto s_a}{s_\bot}}\\
G_j^i := &\begin{cases}
\delta_\1 & \text{if } i = j\\
\delta_\0 & \text{if } i = j + 1\\
\delta_\bullet & \text{else}
\end{cases}
\end{align*}
\[\emptyset \vdash G_\star : \Gamma_m(z_\star) \qquad\quad
\emptyset \vdash G_\0 : \Gamma_m(z_\0) \qquad\quad
\emptyset \vdash G_{ab \Rightarrow cd} : \Gamma_m(r_i) \text{ for } i \in \{1, \ldots, L\}\]
\end{defi}

Similarly to substitutions $S_F$ and $S_H$, we introduce the following substitution $S_G$.

\begin{defi}[Substitution $S_G$]\label{def:S_G}
\coqleft{L834}%
\[S_G(z_\star) := G_\star \qquad
S_G(z_\1) := \pi_\1 \qquad
S_G(z_\0) := G_\0 \qquad
S_G(r_j) := G_{R_j} \text{ for } j \in \{1, \ldots, L\}\]
\end{defi}

The following Example~\ref{xmp:rule} illustrates how the term $G_{ab \Rightarrow cd}$ together with the term $G_j^i$ represent the effect of rule $ab \Rightarrow cd$ applied at position $j$ on the symbol at position $i$.

\begin{exa}\label{xmp:rule}
Consider for symbols $\0, \1, \ldots, \5$ an application of the rule $\1\2 \Rightarrow \4\5$ at position $2$ in order to rewrite the word $\0\1\2\3$ to $\0\4\5\3$.
Accordingly, we have:
\begin{description}
\item[Position $1$] $G_{\1\2 \Rightarrow \4\5}\,G_2^1\,\pi_\0 =_\beta \pi_\0$
\item[Position $2$] $G_{\1\2 \Rightarrow \4\5}\,G_2^2\,\pi_{\4} =_\beta \pi_\1$
\item[Position $3$] $G_{\1\2 \Rightarrow \4\5}\,G_2^3\,\pi_{\5} =_\beta \pi_\2$
\item[Position $4$] $G_{\1\2 \Rightarrow \4\5}\,G_2^4\,\pi_\3 =_\beta \pi_\3$
\end{description}
Let us also inspect an illegal rewriting step with rule $\1\2 \Rightarrow \4\5$ at position~$1$, for which we have $G_{\1\2 \Rightarrow \4\5}\,G_1^1\,\pi_\0 =_\beta \pi_\bot$.
This means the symbol $\0$ cannot be at position $1$ as the result of the application of rule $\1\2 \Rightarrow \4\5$ at position $1$.
\end{exa}

The above observation is generalized for terms in $\calQ_m$ in the following Lemma~\ref{lem:Q_ssts_sem}.

\begin{lem}\label{lem:Q_ssts_sem}
\reversemarginpar\coqleft{L2106}%
For $m > 0$, let $a_1, \ldots, a_{m + 1} \in \{\0, \1, \ldots, \K\}$ and $M \in \calQ_m$. If $\Gamma_m \vdash M : \kappa$,\\
$S_G(M)[p_1 := G_1^0, \ldots, p_m := G_m^0] =_\beta \pi_\1$, and
$S_G(M)[p_1 := G_1^i, \ldots, p_m := G_m^i] =_\beta \pi_{a_i}$ for $i \in \{1, \ldots, m+1\}$,
then $a_1 \ldots a_{m + 1} \Rightarrow_\frakR^* \1^{m+1}$.
\end{lem}

\begin{proof}
Induction on the size of $M$ and case analysis using Definition~\ref{def:ring}.
\end{proof}

The following Example~\ref{xmp:G_Q} builds upon the previous Example~\ref{xmp:H} and illustrates the intended meaning (rewriting $\0$s to $\1$s) of a \enquote{well-formed} example term in $\calQ_3$. 

\begin{exa}\label{xmp:G_Q}
Assume $\frakR = \{\0\0 \Rightarrow \2\2, \0\2 \Rightarrow \1\1, \2\0 \Rightarrow \1\1\}$ over the alphabet $\{\0,\1,\2\}$, and consider the term $N = r_1\,p_2\,(r_2\,p_1\,(r_3\,p_3\,z_\1))$ from Example~\ref{xmp:H}.
Replacing $G_j^i$ accordingly for $i \in \{0, \ldots, 4\}$ and $j \in \{1, 2, 3\}$, we have the following $\beta$-equivalences $(0)$ -- $(4)$.
$\begin{array}{lll}
(0) &\quad &S_G(N)[p_1 := \delta_\bullet, p_2 := \delta_\bullet, p_3 := \delta_\bullet] =_\beta \pi_\1 \\
(1) &\quad &S_G(N)[p_1 := \delta_\1, p_2 := \delta_\bullet, p_3 := \delta_\bullet] =_\beta \pi_\0\\
(2) &\quad &S_G(N)[p_1 := \delta_\0, p_2 := \delta_\1, p_3 := \delta_\bullet] =_\beta \pi_\0\\
(3) &\quad &S_G(N)[p_1 := \delta_\bullet, p_2 := \delta_\0, p_3 := \delta_\1] =_\beta \pi_\0\\
(4) &\quad &S_G(N)[p_1 := \delta_\bullet, p_2 := \delta_\bullet, p_3 := \delta_\0] =_\beta \pi_\0
\end{array}$

In accordance with Lemma~\ref{lem:Q_ssts_sem}, we have that $\0^4 \Rightarrow_\frakR^* \1^4 $.
Equivalences $(1)$ -- $(4)$ witness the particular rewriting steps at positions $1$ -- $4$ (cf.~Example~\ref{xmp:H} and Example~\ref{xmp:rule}).
Illegal rule application (cf.\,Example~\ref{xmp:no-H}) would result in $\pi_\bot$ as one of the right-hand sides.
\end{exa}

Complementarily to word rewriting, the following Lemma~\ref{lem:R_ssts_sem} characterizes word expansion and initialization with $\0$s.

\begin{lem}\label{lem:R_ssts_sem}
\reversemarginpar\coqleft{L2017}%
If $M \in \calR_1$ such that $\Gamma_1 \vdash M : \kappa$,
$S_G(M)[p_1 := \delta_\bullet] =_\beta \pi_\$$,\\
$S_G(M)[p_1 := \delta_\1] =_\beta \pi_\0$, and
$S_G(M)[p_1 := \delta_\0] =_\beta \pi_\1$,
then there exists an $m > 0$ and an~$N \in \calQ_m$ such that $\Gamma_m \vdash N : \kappa$,
$S_G(N)[p_1 := G_1^0, \ldots, p_m := G_m^0] =_\beta \pi_\1$,\\
and $S_G(N)[p_1 := G_1^i, \ldots, p_m := G_m^i] =_\beta \pi_\0$ for $i \in \{1, \ldots, m+1\}$.
\end{lem}

\begin{proof}
Considering the general case $M \in \calR_{m'}$ for $m' > 0$, induction on the size of $M$ and case analysis using Definition~\ref{def:ring}.
\end{proof}

The following Example~\ref{xmp:G_R} complements the previous Example~\ref{xmp:G_Q} and illustrates the intended meaning (word expansion and initialization with $\0$s) of a \enquote{well-formed} term in the set $\calR_1$.

\begin{exa}\label{xmp:G_R}
Assume $\frakR = \{\0\0 \Rightarrow \2\2, \0\2 \Rightarrow \1\1, \2\0 \Rightarrow \1\1\}$ over the alphabet $\{\0,\1,\2\}$, and consider the terms
$M_1 := z_\star\,p_1\,(\lambda p_2.M_2)$, $M_2 := z_\star\,p_2\,(\lambda p_3.M_3)$,
$M_3 := z_\0\,p_3\,N$, and $N := r_1\,p_2\,(r_2\,p_1\,(r_3\,p_3\,z_\1))$ from Example~\ref{xmp:G_Q}.
Proceeding bottom up, we have $M_3 \in \calR_3$ and the following $\beta$-equivalences hold:\\
$\begin{array}{l}
S_G(M_3)[p_1 := \delta_\bullet, p_2 := \delta_\bullet, p_3 := \delta_\bullet] =_\beta \pi_\$\\
S_G(M_3)[p_1 := \delta_\1, p_2 := \delta_\bullet, p_3 := \delta_\bullet] =_\beta \pi_\0\\
S_G(M_3)[p_1 := \delta_\0, p_2 := \delta_\1, p_3 := \delta_\bullet] =_\beta \pi_\0\\
S_G(M_3)[p_1 := \delta_\bullet, p_2 := \delta_\0, p_3 := \delta_\1] =_\beta \pi_\0\\
S_G(M_3)[p_1 := \delta_\bullet, p_2 := \delta_\bullet, p_3 := \delta_\0] =_\beta \pi_\1
\end{array}$\\
Additionally, $M_2 \in \calR_2$, $M_1 \in \calR_1$, and the following $\beta$-equivalences hold:\\
$\begin{array}{l}
S_G(M_2)[p_1 := \delta_\bullet, p_2 := \delta_\bullet] =_\beta \pi_\$\\
S_G(M_2)[p_1 := \delta_\1, p_2 := \delta_\bullet] =_\beta \pi_\0\\
S_G(M_2)[p_1 := \delta_\0, p_2 := \delta_\1] =_\beta \pi_\0\\
S_G(M_2)[p_1 := \delta_\bullet, p_2 := \delta_\0] =_\beta \pi_\1
\end{array}$
\qquad
$\begin{array}{l}
S_G(M_1)[p_1 := \delta_\bullet] =_\beta \pi_\$\\
S_G(M_1)[p_1 := \delta_\1] =_\beta \pi_\0\\
S_G(M_1)[p_1 := \delta_\0] =_\beta \pi_\1
\end{array}$

In combination with the previous Example~\ref{xmp:G_Q}, the term $M_1 \in \calR_1$ represents word expansion up to length $4$, followed by initialization with $\0$s, and rewriting to $\1$s.
\end{exa}

Next, we combine syntactic and semantic constraints in the following key Lemma~\ref{lem:ssts_to_constr}.

\begin{lem}\label{lem:ssts_to_constr}
\coqleft{L2758}%
There exists an $n \in \bbN$ such that $\0^{n+1} \Rightarrow_\frakR^* \1^{n+1}$ iff
there exists a term $M$ in normal form, such that the following conditions hold:\\
$\emptyset \vdash M : \Gamma_1(r_1) \to \cdots \to \Gamma_1(r_L) \to \Gamma_1(z_\0) \to \Gamma_1(z_\1) \to \Gamma_1(z_\star) \to \Gamma_1(p_1) \to \kappa$,\\
$\begin{array}{ll}
M\,I \ldots I\,(\lambda h.I)\,u\,(\lambda h.\lambda g.g\,I)\,I &=_\beta u,\\
M\,H_R \ldots H_R\,H_\0\,\pi_\1\,H_\star\,\delta_\bullet &=_\beta \pi_{\$},\\
M\,G_{R_1} \ldots G_{R_L}\,G_\0\,\pi_\1\,G_\star\,\delta_\bullet &=_\beta \pi_{\$},\\
M\,G_{R_1} \ldots G_{R_L}\,G_\0\,\pi_\1\,G_\star\,\delta_\1 &=_\beta \pi_{\0},\\
M\,G_{R_1} \ldots G_{R_L}\,G_\0\,\pi_\1\,G_\star\,\delta_\0 &=_\beta \pi_{\1}.
\end{array}$
\end{lem}

\begin{proof}
The direction from left to right proceeds in two steps.
First, by induction on the number of rewriting steps we construct a term $N \in \calQ_n$ (easy converse of Lemma~\ref{lem:Q_ssts_sem}).
Second, by induction on $n$ we construct a term $M' \in \calR_1$ containing $N \in \calQ_n$ as a subterm (easy converse of Lemma~\ref{lem:R_ssts_sem}).
Then, the solution is $M := \lambda r_1 \ldots r_L.\lambda z_\0 z_\1 z_\star p_1.M'$.

The direction from right to left proceeds in two steps.
First, by Theorem~\ref{thm:shape} we have $M = \lambda r_1 \ldots r_L.\lambda z_\0 z_\1 z_\star p_1.M'$ for some $M' \in \calR_1$.
Second, by Lemma~\ref{lem:R_ssts_sem} and Lemma~\ref{lem:Q_ssts_sem} we have $\0^{n+1} \Rightarrow_\frakR^* \1^{n+1}$ for some $n \in \bbN$.
\end{proof}

Finally, we present the combination of constraints from the above Lemma~\ref{lem:ssts_to_constr} as a matching instance $F_\frakR\,\mathsf{X} =_\beta N_\frakR$.
This constitutes the main result of the present work.

\begin{thm}\label{thm:ssts_to_hom}
\coqleft{L2880}%
Problem $\0^+ \Rightarrow^* \1^+$ many-one reduces to higher-order $\beta$-matching.
\end{thm}

\begin{proof}
Given a simple semi-Thue system $\frakR = \{R_1, \ldots, R_L\}$ due to Lemma~\ref{lem:ssts_to_constr} there exists an $n \in \bbN$ such that $\0^{n+1} \Rightarrow_\frakR^* \1^{n+1}$ iff the instance $F_\frakR\,\mathsf{X} =_\beta N_\frakR$ of higher-order $\beta$-matching is solvable, where
\begin{itemize}
\item $F_\frakR := \lambda x.\lambda y.y\,\begin{array}[t]{@{}l}
(\lambda u.x \underbracket{I \ldots I}_{L \text{ times}}\,(\lambda h.I)\,u\,(\lambda h.\lambda g.g\,I)\,I)\\
(x \underbracket{H_R \ldots H_R}_{L \text{ times}}\,H_\0\,\pi_\1\,H_\star\,\delta_\bullet)\\
(x\,G_{R_1} \ldots G_{R_L}\,G_\0\,\pi_\1\,G_\star\,\delta_\bullet)\\
(x\,G_{R_1} \ldots G_{R_L}\,G_\0\,\pi_\1\,G_\star\,\delta_\1)\\
(x\,G_{R_1} \ldots G_{R_L}\,G_\0\,\pi_\1\,G_\star\,\delta_\0)
\end{array}$
\item $N_\frakR := \lambda y.y\,(\lambda u.u)\,\pi_\$\,\pi_\$\,\pi_\0\,\pi_\1$
\item $\sigma_\frakR := \Gamma_1(r_1) \to \cdots \to \Gamma_1(r_L) \to \Gamma_1(z_\0) \to \Gamma_1(z_\1) \to \Gamma_1(z_\star) \to \Gamma_1(p_1) \to \kappa$
\item $\tau_\frakR := ((\kappa \to \kappa) \to \kappa \to \kappa \to \kappa \to \kappa \to \ga) \to \ga$
\item $\emptyset \vdash F_\frakR : \sigma_\frakR \to \tau_\frakR$
\item $\emptyset \vdash N_\frakR : \tau_\frakR$\qedhere
\end{itemize}
\end{proof}

\begin{thm}\label{thm:hom_undec}
\coqextraleft{LambdaCalculus/HOMatching\_undec.v\#L16}%
Higher-order $\beta$-matching (Problem~\ref{prb:hom}) is undecidable.
\end{thm}

\begin{proof}
By reduction from the undecidable Problem $\0^+ \Rightarrow^* \1^+$ (Theorem~\ref{thm:ssts_undec} and Theorem~\ref{thm:ssts_to_hom}).
\end{proof}

\begin{rem}\label{rem:order}
The \emph{order} of a type is the maximal nesting depth of the arrow type constructor to the left, starting by $\text{order}(\ga) = 1$.
In the proof of the above Theorem~\ref{thm:ssts_to_hom} we have $\text{order}(\sigma_\frakR) = 1 + \text{order}(\Gamma_1(z_\star)) = 6$.
This agrees with Loader's result that $\beta$-matching is undecidable at order $6$.
\end{rem}

\section{Mechanization}\label{sec:mech}

This section provides a brief overview over the mechanization of undecidability of higher-order $\beta$-matching (Theorem~\ref{thm:hom_undec}) using the Rocq Prover~\cite{Coq_2023}.
The mechanization is axiom-free and spans approximately $4000$ lines of code, consisting of the following parts:
\begin{itemize}
\item \texttt{HOMatching.v} contains definitions of the simply typed $\lambda$-calculus~\coqextra{LambdaCalculus/HOMatching.v\#L31} and higher-order $\beta$-matching~\coqextra{LambdaCalculus/HOMatching.v\#L37}.
\item \texttt{Util/stlc\_facts.v} and \texttt{Util/term\_facts.v} contain basic properties of the simply typed $\lambda$-calculus, such as confluence of $\beta$-reduction~\coqextra{LambdaCalculus/Util/term\_facts.v\#L1056}, substitution lemmas~\coqextra{LambdaCalculus/Util/stlc\_facts.v\#L100}, and type preservation properties~\coqextra{LambdaCalculus/Util/stlc\_facts.v\#L116}.
\item \texttt{Reductions/SSTS01\_to\_HOMbeta.v} contains the reduction from Problem $\0^+ \Rightarrow^* \1^+$ to higher-order $\beta$-matching~\coq{L2880}.
\item \texttt{HOMatching\_undec.v} contains the undecidability result for higher-order $\beta$-matching~\coqextra{LambdaCalculus/HOMatching\_undec.v\#L16}.
\end{itemize}

The simple type system \texttt{stlc} is mechanized in \texttt{HOMatching.v}, borrowing the existing term definitions from the library~\coqextra{L/L.v\#L8}, for which variable binding is addressed via the unscoped de Bruijn approach~\cite{deBruijn72}.
The proposition \texttt{stlc~Gamma~M~t} mechanizes that the term \texttt{M} is assigned the simple type~\texttt{t} in the simple type environment \texttt{Gamma}.
\begin{lstlisting}[mathescape]
Inductive ty : Type :=
  | atom (* type variable *)
  | arr (s t : ty). (* function type *)

Inductive term : Type :=
  | var (n : nat) : term (* term variable *)
  | app (s : term) (t : term) : term (* application *)
  | lam (s : term). (* abstraction *)

Inductive stlc (Gamma : list ty) : term -> ty -> Prop :=
  | stlc_var x t : nth_error Gamma x = Some t ->
      stlc Gamma (var x) t (* variable rule *)
  | stlc_app M N s t : stlc Gamma M (arr s t) -> stlc Gamma N s ->
      stlc Gamma (app M N) t (* application rule *)
  | stlc_lam M s t : stlc (cons s Gamma) M t ->
      stlc Gamma (lam M) (arr s t). (* abstraction rule *)
\end{lstlisting}

Higher-order $\beta$-matching is mechanized as the predicate \texttt{HOMbeta}: given terms \texttt{F} of type \texttt{arr s t} and \texttt{N} of type \texttt{t}, is there a simply typed term \texttt{M} of type \texttt{s} such that \texttt{app F M} is $\beta$-equivalent (reflexive, symmetric, transitive closure of \texttt{step}) to \texttt{N}?
\begin{lstlisting}[mathescape]
Definition HOMbeta : { '(s, t, F, N) : (ty * ty * term * term) 
  | stlc nil F (arr s t) /\ stlc nil N t } -> Prop :=
    fun '(exist _ (s, t, F, N) _) =>
      exists (M : term), stlc nil M s /\
 clos_refl_sym_trans term step (app F M) N.
\end{lstlisting}

The proposition \lstinline|undecidable HOMbeta|~\coqextra{LambdaCalculus/HOMatching\_undec.v\#L16} mechanizes the undecidability of the predicate \texttt{HOMbeta}, relying on the following library definition~\cite[Chapter~19]{forster2021computability}.
A predicate \texttt{p} is undecidable, if existence of a computable decider for \texttt{p} implies recursive co-enumerability of the (Turing machine) Halting Problem.
\begin{lstlisting}[mathescape]
Definition undecidable {X} (p : X -> Prop) :=
 decidable p -> enumerable (complement SBTM_HALT).
\end{lstlisting}
Since the Halting Problem is recursively enumerable, decidability of \texttt{p} would imply decidability of the Halting Problem.

\section{On Intersection Type Inhabitation and $\lambda$-Definability}\label{sec:uniform}

We conclude the technical presentation with the following observation:
the presented approach reducing Problem $\0^+ \Rightarrow^* \1^+$ to higher-order $\beta$-matching is easily transferred to intersection type inhabitation and $\lambda$-definability.

The following Remark~\ref{rem:G} shows the structure of the corresponding finite model with respect to the present construction.

\begin{rem}\label{rem:G}
Terms in Definition~\ref{def:G} realize certain finite functions as follows.
\begin{itemize}
\item $\delta_i$ for $i \in \calA$ realizes a member of the finite function family specified by the partial function table ${\setlength{\arraycolsep}{1pt}\left(\begin{array}{l}
\top \mapsto i\end{array}\right)}$.
\item $G_\0$ realizes a member of the family specified by ${\left(\begin{array}{l}
(\top \mapsto \bullet) \mapsto (\0 \mapsto \0)\\
(\top \mapsto \bullet) \mapsto (\1 \mapsto \$)\\
(\top \mapsto \0) \mapsto (\0 \mapsto \1)\\
(\top \mapsto \1) \mapsto (\0 \mapsto \0)
\end{array}\right)}$.
\item $G_{ab \Rightarrow cd}$ realizes a member of the family specified by ${\left(\begin{array}{l}
(\top \mapsto \1) \mapsto (c \mapsto a)\\
(\top \mapsto \0) \mapsto (d \mapsto b)
\end{array}\right)}$.
\item $G_\star$ realizes a member of the family specified by ${\setlength{\arraycolsep}{1pt}\left(\begin{array}{l}
(\top \mapsto \bullet) \mapsto (((\top \mapsto \bullet) \mapsto \0) \mapsto \0)\\
(\top \mapsto \bullet) \mapsto \left(\left(\begin{array}{l}
(\top \mapsto \bullet) \mapsto \$\\
(\top \mapsto \0) \mapsto \1
\end{array}\right) \mapsto \$\right)\\
(\top \mapsto \0) \mapsto (((\top \mapsto \1) \mapsto \0) \mapsto \1)\\
(\top \mapsto \1) \mapsto (((\top \mapsto \bullet) \mapsto \0) \mapsto \0)
\end{array}\right)}$.
\end{itemize}
\end{rem}

The above specifications follow the intended meaning (Example~\ref{xmp:H}) of the corresponding programs, when used in \enquote{well-formed} terms in $\calQ_m$ and $\calR_m$.
For example, we have $G_\0\,\delta_\bullet\,\pi_\1 =_\beta \pi_\$$, in agreement with the above Remark~\ref{rem:G}.

Let us briefly recapitulate one standard presentation of the intersection type system~\cite{Bakel11}.
Intersection types extend simple types with an associative, commutative, and idempotent constructor $\cap$.
Judgments are of shape $\Gamma \vdash_\cap M : \varphi$, where the intersection type environment $\Gamma$ contains assumptions of shape $(x : \varphi_1 \cap \cdots \cap \varphi_n)$ where $n \geq 0$.

\begin{defi}[Intersection Types]
\label{def:itypes}~
\[
\varphi ::= a \mid \varphi_1 \cap \cdots \cap \varphi_n \to \varphi
\]
where $n \geq 0$ and $a$ ranges over a denumerable set of type variables.
\end{defi}

\begin{defi}[Intersection Type System~{\cite{Bakel11}}]~\medskip\\
\begin{minipage}{0.9\textwidth}
\centering
\begin{tabular}{l}
{\RightLabel{\textnormal{(Ax)}}
\AxiomC{$i \in \{1, \ldots, n\}$ }
\UnaryInfC{$\Gamma, x : \varphi_1 \cap \cdots \cap \varphi_n \vdash_\cap x : \varphi_i$}
\DisplayProof}\quad
{\RightLabel{\textnormal{($\to$I)}}
\AxiomC{$\Gamma, x: \varphi_1 \cap \cdots \cap \varphi_n \vdash_\cap M : \varphi$}
\UnaryInfC{$\Gamma \vdash_\cap \lambda x.M : \varphi_1 \cap \cdots \cap \varphi_n \to \varphi$}
\DisplayProof}
\\\\
{\RightLabel{\textnormal{($\to$E)}}
\AxiomC{$\Gamma \vdash_\cap M : \varphi_1 \cap \cdots \cap \varphi_n \to \varphi$}
\AxiomC{$\Gamma \vdash_\cap N : \varphi_1,\quad\ldots,\quad \Gamma \vdash_\cap N : \varphi_n$}
\BinaryInfC{$\Gamma \vdash_\cap M\,N : \varphi$}
\DisplayProof}
\end{tabular}
\end{minipage}
\end{defi}

The above type system enjoys subject reduction with respect to $\beta$-reduction and strong normalization properties~\cite{Barendregt13book} in the absence of the empty intersection.

The decision problem of intersection type inhabitation is: given an intersection type environment $\Gamma$ and an intersection type $\varphi$, is there a term $M$ such that $\Gamma \vdash_\cap M : \varphi$ is derivable?
Undecidability of intersection type inhabitation was shown by Urzyczyn~\cite{Urzyczyn99}.

\newpage

Let us illustrate in the following Example~\ref{xmp:inh_Q} and Example~\ref{xmp:inh_R} the relationship between intersection type inhabitation and the $\beta$-equivalence construction in Example~\ref{xmp:G_Q} and Example~\ref{xmp:G_R} respectively.

\begin{exa}\label{xmp:inh_Q}
Assume $\frakR = \{\0\0 \Rightarrow \2\2, \0\2 \Rightarrow \1\1, \2\0 \Rightarrow \1\1\}$ over the alphabet $\{\0,\1,\2\}$, and consider the term $N = r_1\,p_2\,(r_2\,p_1\,(r_3\,p_3\,z_\1))$ from Example~\ref{xmp:H}. Let
\begin{align*}
\Gamma := \{&r_1 : (\1 \to \2 \to \0) \cap (\0 \to \2 \to \0) \cap (\bullet \to \0 \to \0) \cap (\bullet \to \1 \to \1) \cap (\bullet \to \2 \to \2),\\
&r_2 : (\1 \to \1 \to \0) \cap (\0 \to \1 \to \2) \cap (\bullet \to \0 \to \0) \cap (\bullet \to \1 \to \1) \cap (\bullet \to \2 \to \2),\\
&r_3 : (\1 \to \1 \to \2) \cap (\0 \to \1 \to \0) \cap (\bullet \to \0 \to \0) \cap (\bullet \to \1 \to \1) \cap (\bullet \to \2 \to \2),\\
&z_\1 : \1\}
\end{align*}
The following judgments $(0)$ -- $(4)$ are derivable:\\
$\begin{array}{lll}
(0) &\quad &\Gamma \cup \{ p_1 : \bullet, p_2 : \bullet, p_3 : \bullet \} \vdash_\cap N : \1 \\
(1) &\quad &\Gamma \cup \{ p_1 : \1, p_2 : \bullet, p_3 : \bullet \} \vdash_\cap N : \0\\
(2) &\quad &\Gamma \cup \{ p_1 : \0, p_2 : \1, p_3 : \bullet \} \vdash_\cap N : \0\\
(3) &\quad &\Gamma \cup \{ p_1 : \bullet, p_2 : \0, p_3 : \1 \} \vdash_\cap N : \0\\
(4) &\quad &\Gamma \cup \{ p_1 : \bullet, p_2 : \bullet, p_3 : \0 \} \vdash_\cap N : \0
\end{array}$

The above judgments are in one-to-one correspondence to $\beta$-equivalences $(0)$ -- $(4)$ from Example~\ref{xmp:G_Q} and witness the particular rewriting steps at positions $1$ -- $4$ (cf.~Example~\ref{xmp:H} and Example~\ref{xmp:rule}).
\end{exa}

\begin{exa}\label{xmp:inh_R}
Assume $\frakR = \{\0\0 \Rightarrow \2\2, \0\2 \Rightarrow \1\1, \2\0 \Rightarrow \1\1\}$ over the alphabet $\{\0,\1,\2\}$, and consider the terms
$M_1 := z_\star\,p_1\,(\lambda p_2.M_2)$, $M_2 := z_\star\,p_2\,(\lambda p_3.M_3)$,
$M_3 := z_\0\,p_3\,N$, and $N := r_1\,p_2\,(r_2\,p_1\,(r_3\,p_3\,z_\1))$ from Example~\ref{xmp:G_Q}.
Consider $\Gamma$ from Example~\ref{xmp:inh_Q} and let
\begin{align*}
\Gamma' := \{&z_\0 : (\bullet \to \0 \to \0) \cap (\bullet \to \1 \to \$) \cap (\0 \to \0 \to \1) \cap (\1 \to \0 \to \0),\\
&z_\star : (\bullet \to (\bullet \to \0) \to \0) \cap (\bullet \to ((\bullet \to \$) \cap (\0 \to \1)) \to \$) \,\cap \\
&\hspace*{1.7em}(\0 \to (\1 \to \0) \to \1) \cap (\1 \to (\bullet \to \0) \to \0)\}
\end{align*}

Proceeding bottom up, the following judgments are derivable:\medskip\\
$\begin{array}{l}
\Gamma \cup \Gamma' \cup \{p_1 : \bullet, p_2 : \bullet, p_3 : \bullet\} \vdash_\cap M_3 : \$\\
\Gamma \cup \Gamma' \cup \{p_1 : \1, p_2 : \bullet, p_3 : \bullet\} \vdash_\cap M_3 : \0\\
\Gamma \cup \Gamma' \cup \{p_1 : \0, p_2 : \1, p_3 : \bullet\} \vdash_\cap M_3 : \0\\
\Gamma \cup \Gamma' \cup \{p_1 : \bullet, p_2 : \0, p_3 : \1\} \vdash_\cap M_3 : \0\\
\Gamma \cup \Gamma' \cup \{p_1 : \bullet, p_2 : \bullet, p_3 : \0\} \vdash_\cap M_3 : \1
\end{array}$\medskip\\
Additionally, the following judgments are derivable:\medskip\\
$\begin{array}{l}
\Gamma \cup \Gamma' \cup \{p_1 : \bullet, p_2 : \bullet\} \vdash_\cap M_2 : \$\\
\Gamma \cup \Gamma' \cup \{p_1 : \1, p_2 : \bullet\} \vdash_\cap M_2 : \0\\
\Gamma \cup \Gamma' \cup \{p_1 : \0, p_2 : \1\} \vdash_\cap M_2 : \0\\
\Gamma \cup \Gamma' \cup \{p_1 : \bullet, p_2 : \0\} \vdash_\cap M_2 : \1
\end{array}$
\qquad
$\begin{array}{l}
\Gamma \cup \Gamma' \cup \{p_1 : \bullet\} \vdash_\cap M_1 : \$\\
\Gamma \cup \Gamma' \cup \{p_1 : \1\} \vdash_\cap M_1 : \0\\
\Gamma \cup \Gamma' \cup \{p_1 : \0\} \vdash_\cap M_1 : \1
\end{array}$

The above judgments are in one-to-one correspondence to $\beta$-equivalences from Example~\ref{xmp:G_R} and witness word expansion up to length $4$, followed by initialization with $\0$s, and rewriting to $\1$s.
\end{exa}

A construction, similar to the above Examples is carried out and mechanized\;\coqextra{IntersectionTypes/Reductions/SSTS01\_to\_CD\_INH.v} systematically~\cite{DudenhefnerR19} as a refinement of Urzyczyn's original undecidability proof.

Finally, let us state the relationship between simple semi-Thue system rewriting, higher-order $\beta$-matching, $\lambda$-definability, and intersection type inhabitation in the following Proposition~\ref{prop:uniform}.

\begin{prop}\label{prop:uniform}
Given a simple semi-Thue system $\frakR$, one can construct simply typed terms $F_\frakR$ and $N_\frakR$, an intersection type $T_\frakR$, and a finite function $\calF_\frakR$ such that the following statements are equivalent:
\begin{enumerate}
\item There exists an $n \in \bbN$ such that $\0^{n+1} \Rightarrow_\frakR^* \1^{n+1}$.
\item The instance $F_\frakR\,\mathsf{X} =_\beta N_\frakR$ of higher-order $\beta$-matching is solvable.
\item The intersection type $T_\frakR$ is inhabited.
\item The finite function $\calF_\frakR$ is $\lambda$-definable.
\end{enumerate}
\end{prop}

The presented approach shows $(1) \iff (2)$.
With some effort, $(1) \iff (3)$ can be concluded via the approach showing undecidability of intersection type inhabitation~\cite{Urzyczyn09}.
Similarly, $(1) \iff (4)$ can be obtained via a modification of Loader's proof for undecidability of $\lambda$-definability~\cite{loader2001undecidability}, taking into account the correspondence of Salvati et al.~\cite{SalvatiMGB12}.
However, we make the following two observations regarding an alternative, uniform argument.
First, based on Remark~\ref{rem:R1}, the approach is easily adapted to show $(1) \iff (3)$, such that the inhabitant is essentially a member of $\calR_1$.
This is already done in the existing mechanized reduction from Problem $\0^+ \Rightarrow^* \1^+$ to intersection type inhabitation\;\coqextra{IntersectionTypes/Reductions/SSTS01\_to\_CD\_INH.v}.
Second, based on Remark~\ref{rem:G}, the approach can be adapted to show $(1) \iff (4)$, such that the realizer is essentially a member of $\calR_1$.
This is further supported by the known correspondence between intersection type inhabitation (in the fragment at hand) and $\lambda$-definability~\cite{SalvatiMGB12}.

\section{Conclusion}\label{sec:concl}
The present work presents a new, mechanized proof of the undecidability of higher-order $\beta$-matching.
The mechanization is contributed to the existing Coq Library of Undecidability Proofs~\cite{CLUP20}.

While the existing proofs by Loader~\cite{Loader03} and by Joly~\cite{Joly05} are each based on variants of $\lambda$-definability, the presented proof reduces a rewriting problem (Problem $\0^+ \Rightarrow^* \1^+$) to higher-order $\beta$-matching.
As a result, the proof is simpler to verify in full detail and yields a concise mechanization.
Additionally, undecidability of Problem $\0^+ \Rightarrow^* \1^+$ is already mechanized, and is part of the Coq Library of Undecidability Proofs.

Besides the main technical result, we argue that the present approach is uniformly applicable to show undecidability of intersection type inhabitation and $\lambda$-definability.
The former is already established and implemented as refinement~\cite{DudenhefnerR19} of Urzyczyn's undecidability result~\cite{Urzyczyn09}.
The latter is an application of the known correspondence between intersection type inhabitation and $\lambda$-definability~\cite{SalvatiMGB12}.
A mechanization of this argument remains open.

By Remark~\ref{rem:order}, the present approach agrees with Loader's result that $\beta$-matching is undecidable at order $6$.
While Loader conjectures that order $5$ may suffice, neither Loader's technique (as observed by Joly~\cite[Section~5]{Joly05}), nor the present approach is applicable at order $5$.
Constraining the shape of candidate solutions both in the present work as well as in Loader's proof seems to necessitate order $6$.
While $\beta$-matching at order $4$ is decidable~\cite{Padovani00}, decidability at order $5$ remains an open question.

As pointed out in Remark~\ref{rem:typing-break}, the presented approach might be adapted to scenarios beyond the simply typed $\lambda$-calculus.
An interesting alternative to the simple type system is the Coppo-Dezani intersection type assignment system~\cite{CoppoDezani80}, which characterizes strong normalization~\cite{amadio1998domains}.
Well-typedness in this system would allow for more solution candidates and require more effort with respect to syntactic constraints (cf.\,Section~\ref{sec:undec}).
It is reasonable to believe that higher-order $\beta$-matching is undecidable in any type system for the $\lambda$-calculus which includes the simple type system.

The present mechanization was completed two decades after Loader's proof and is based on two developments.
First, the synthetic approach to computability~\cite{forster2021computability} realized in the undecidability library~\cite{CLUP20} provides a suitable framework for mechanized undecidability results.
The existing infrastructure for the {$\lambda$-calculus} served as an excellent starting point for the development.
Second, proof automation based on \texttt{auto} and \texttt{lia} tactics, and the ability to adapt a mechanization to an evolving construction supported the conception of the present approach both at the intuitive and at the technical level.
While proofs of the individual lemmas (cf.\,Section~\ref{sec:undec}) in the development are quite simplistic, they involve exhaustive case analyses and are sensitive to the exact details of the underlying construction.
Once all cases are covered, there is little room for doubt that the construction is correct.
As a result, the proof was developed via interaction with the proof assistant prior to being transcribed into a traditional written format.

\bibliographystyle{alphaurl}
\bibliography{bibliography}
\end{document}